\documentclass[11pt]{article}

\usepackage{amsmath, amssymb, amsthm, aligned-overset} 
\usepackage{algorithm, algorithmicx}
\usepackage{float}
\usepackage{comment}
\usepackage{dsfont}

\usepackage{tikz}

\allowdisplaybreaks
\expandafter\let\expandafter\oldproof\csname\string\proof\endcsname
\let\oldendproof\endproof
\renewenvironment{proof}[1][\proofname]{%
	\oldproof[\bf #1]%
}{\oldendproof}

\allowdisplaybreaks

\theoremstyle{plain}
\newtheorem{theorem}{Theorem}
\newtheorem{lemma}{Lemma}[section]
\newtheorem{claim}[lemma]{Claim}
\newtheorem{fact}[lemma]{Fact}

\newtheorem{remark}[lemma]{Remark}
\newtheorem{definition}[lemma]{Definition}

\newcommand{\hyp}{\mathcal H}
\newcommand{\vex}{V(\hyp)}
\newcommand{\edg}{E(\hyp)}

\newcommand{\T}{\mathcal T}

\newcommand{\X}{\mathcal X}

\newcommand{\poly}{\text{poly}}
\RequirePackage[normalem]{ulem} 
\RequirePackage{color}\definecolor{RED}{rgb}{1,0,0}\definecolor{BLUE}{rgb}{0,0,1} 

\usepackage{hyperref}

\begin{document}

\title{A Fast Coloring Oracle for Average Case Hypergraphs}

\author{Cassandra Marcussen\thanks{School of Engineering and Applied Sciences, Harvard University, Cambridge, Massachusetts, USA. Email: cmarcussen@g.harvard.edu. Supported in part by an NDSEG fellowship, and by NSF Award 2152413 and a Simons Investigator Award to Madhu Sudan.} \and Edward Pyne\thanks{Department of Electrical Engineering and Computer Science, MIT, Cambridge, Massachusetts, USA. Email: epyne@mit.edu. Supported by an NSF Graduate Research Fellowship.} \and Ronitt Rubinfeld\thanks{Computer Science and Artificial Intelligence Laboratory, MIT, Cambridge, Massachusetts, USA. Email: ronitt@csail.mit.edu. Supported by the NSF TRIPODS program (award DMS-2022448)
and CCF-2310818.}  \and Asaf Shapira\thanks{School of Mathematics, Tel Aviv University, Tel Aviv 69978, Israel. Email: asafico@tau.ac.il. Supported in part
by ERC Consolidator Grant 863438.} \and Shlomo Tauber \thanks{School of Computer Science, Tel Aviv University, Tel Aviv 69978, Israel. Email: shlomotauber@mail.tau.ac.il. Supported in part
by ERC Consolidator Grant 863438.}}

\date{}
\maketitle

\begin{abstract}
Hypergraph $2$-colorability is one of the classical NP-hard problems. Person and Schacht [SODA'09] designed a deterministic algorithm whose expected running time is polynomial over a uniformly chosen $2$-colorable $3$-uniform hypergraph. Lee, Molla, and Nagle recently extended this to $k$-uniform hypergraphs for all $k\geq 3$. Both papers relied heavily on the regularity lemma, hence their analysis was involved and their running time hid tower-type constants.

Our first result in this paper is a new simple and elementary deterministic $2$-coloring algorithm that reproves the theorems of Person--Schacht and Lee--Molla--Nagle while avoiding the use of the regularity lemma. We also show how to turn our new algorithm into a randomized one with average expected running time of only $O(n)$.

Our second and main result gives what we consider to be the ultimate evidence of just how easy it is to find a $2$-coloring of an average $2$-colorable hypergraph. We define a {\em coloring oracle} to be an algorithm which, given vertex $v$, assigns color red/blue to $v$ while inspecting as few edges as possible, so that the answers to any sequence of queries to the oracle are consistent with a single legal $2$-coloring of the input. Surprisingly, we show
that there is a coloring oracle that, on average, can answer {\bf every} vertex query in time $O(1)$. 
\end{abstract}

\thispagestyle{empty}

\newpage

\setcounter{page}{1}

\section{Introduction}
Graph partition problems are among the most well-studied topics in algorithmic graph theory.
These problems ask if a graph can be partitioned so that a certain \textit{global} property holds.
Among these properties, probably the most well-studied one is graph and hypergraph colorability.
Let us quickly recall the relevant definitions.
A $k$-uniform hypergraph $\hyp=(\vex, \edg)$ consists of a vertex set $\vex$ and edge set $\edg$ where every edge $e\in\edg$ is a subset of $\vex$ of size $k$. For brevity, we call a $k$-uniform hypergraph a $k$-graph, noting that a $2$-graph coincides with the usual notion of a graph.
We denote $\big|\vex\big| = n$. A $k$-graph $\hyp$ is $d$-colorable if $\vex$ has a $d$-coloring so that in each edge at least two vertices
are colored differently. When a $k$-graph is $2$-colorable, we will also call it {\em bipartite\footnote{In extremal combinatorics, $2$-colorability is also called {\em Property B}.}}. While graph $d$-colorability is NP-hard for $d \geq 3$, Lov\'asz \cite{lovasz_coloring_hypergraph} famously proved that for $k\geq 3$ deciding if a $k$-graph is $2$-colorable is also NP-hard.
Hypergraph 2-colorability has been extensively studied in combinatorics \cite{Czumaj2000ColoringNH, LovszProblemsAR, 743519}, with notable contributions leading to the development of the renowned Lovász Local Lemma \cite{LovszProblemsAR}. In computer science, hypergraph coloring has received significant attention due to its strong connections with fundamental problems in graph coloring and satisfiability of Boolean formulas \cite{10.1145/335305.335310, Lu1998DeterministicHC}. By leveraging approximation techniques from graph coloring, several works \cite{10.5555/763878.763887, 10.5555/645587.659591, KRIVELEVICH200199, KRIVELEVICH20032} have proposed algorithms for properly coloring 2-colorable hypergraphs, where the number of colors depends on the size of the hypergraph.

Given the hardness of deciding graph and hypergraph colorability problems, it is natural to ask if these problems are easy
on average. This question is a natural one in the study of average-case complexity of NP-hard problems. The first result in this direction was obtained by Turner \cite{Turner} who proved that there is an algorithm that can
find a $d$-coloring of almost all $d$-colorable graphs. Note that this does {\em not} give an algorithm whose average running time is polynomial over the set of $d$-colorable graphs.
Such a result was obtained in a highly influential paper of Dyer and Frieze \cite{reduction}. It is of course natural to design such algorithms for $k$-graphs. The first to so were Person and Schacht \cite{PersonSchacht1,PersonPolynomialAlog} who gave an average case polynomial time
algorithm for $3$-graph $2$-colorability. Their result was recently extended to arbitrary $k \geq 3$ by Lee, Molla, and Nagle \cite{lee2024twocoloringbipartiteuniformhypergraphs} who gave an average case $O(n^k)$ algorithm.
The algorithms of \cite{lee2024twocoloringbipartiteuniformhypergraphs,PersonSchacht1,PersonPolynomialAlog} rely on graph/hypergraph regularity lemmas. As a result, they were difficult to analyze, used algorithmic versions of the regularity lemma~\cite{ADLRY} as a black box, and their running times hid tower-type constants.

\subsection{New Classical Algorithms}

Our main results in this paper improve upon \cite{lee2024twocoloringbipartiteuniformhypergraphs,PersonSchacht1,PersonPolynomialAlog} in several ways. In what follows, a $2$-coloring algorithm is one that is guaranteed to find a $2$-coloring of every $2$-colorable $k$-graph.
Throughout the rest of the paper we assume that $k \geq 3$ is an absolute constant. Therefore, while in many places the dependence on $k$ can be improved significantly, we chose not to do so for the sake of simplifying the proofs.

Our first result in this paper is a new deterministic $2$-coloring algorithm with polynomial average case running time. The algorithm is simple and completely elementary. In particular, it does not use any graph or hypergraph regularity lemma.

\begin{theorem}\label{algorithm theorem}
There is a deterministic $2$-coloring algorithm whose average case running time over the set of $2$-colorable $k$-graphs is $n^{O(k)}$.
\end{theorem}

To be precise, and to set the stage for the introduction of the coloring oracle of Theorem \ref{main theorem} below, we state the above theorem as follows.
If $A$ is a deterministic $2$-coloring algorithm, then $T_A(\hyp)$ denotes the running time of $A$ on input $\hyp$, and $T_A(n)$ denotes the average of $T_A(\hyp)$ over all $2$-colorable $k$-graphs on $n$ vertices. Theorem \ref{algorithm theorem} thus states that there is a $2$-coloring algorithm $A$ satisfying $T_A(n)=n^{O(k)}$.

Our second result shows that if we allow the $2$-coloring algorithm to use randomization, then we can significantly improve the average case expected running time.

\begin{theorem}\label{randomized algorithm}
There is a randomized $2$-coloring algorithm whose average case expected running time over the set of $2$-colorable $k$-graphs is $O(n)$.
\end{theorem}

If $A$ is a randomized $2$-coloring algorithm, then $T_A(\hyp)$ now denotes the {\em expected} running time of $A$ on input $\hyp$, and $T_A(n)$ denotes the average of $T_A(\hyp)$ over all
$2$-colorable $k$-graphs on $n$ vertices. Theorem \ref{randomized algorithm} thus states that there is a randomized $2$-coloring algorithm $A$
satisfying $T_A(n)=O(n)$. This average case running time is optimal since we need $O(n)$ time just to output the coloring of $\hyp$.

\subsection{A Fast Coloring Oracle}

We now turn to introducing our main result in this paper, which shows that hypergraph $2$-colorability is even easier on average than merely being solvable in polynomial time as in Theorems \ref{algorithm theorem} and \ref{randomized algorithm}. (In fact, the algorithms and techniques for this result directly imply Theorems \ref{algorithm theorem} and \ref{randomized algorithm} as well.) To this end we will introduce a {\em local} coloring algorithm in Definition
\ref{deforacle} below. Our main inspiration for this definition are the notion of a {\em Partition Oracles} from the area of sublinear time algorithms and the emerging area of {\em Local Computing Algorithms (LCA)}. Let us now describe Partition Oracles, postponing to Subsection \ref{subseclca} the discussion of LCAs and prior relevant work. The notion of partition oracle was first introduced in \cite{Hassidim2009LocalGP} and was further studied in \cite{improve_oracle,kumar2021randomwalksforbiddenminors}. These are sublinear, in fact $O(1)$, time algorithms that supply oracle access to a vertex partition of every planar graph\footnote{The results actually hold for more general ``minor closed'' families of graphs.} so that every component is of size $O(1)$ and there are $o(n)$ edges connecting these components. As is the case in sublinear time algorithms, we assume that the algorithm can quickly query the input object. In our case,
the algorithm can query whether any $k$-tuple of vertices is an edge in $\hyp$.
Given the discussion above, we introduce the following definition.

\begin{definition}[Coloring-Oracle]\label{deforacle}
A coloring-oracle is a randomized algorithm $A$ whose input is a vertex $u$ in a $2$-colorable $k$-graph $\hyp$.
The algorithm $A$ can access $\hyp$ using edge queries of the form ``is $(v_1,\ldots,v_k)$ an edge in $\hyp$''?
The algorithm should always return a color $0/1$ for $u$. Furthermore:
\begin{itemize}
\item[$(i)$] For every sequence of queries $u_1,u_2\ldots$, the answers of $A$ are consistent with a single legal $2$-coloring of $\hyp$.
\item[$(ii)$] The oracle uses only $o(n)$ memory for keeping information between different calls.
\end{itemize}
\end{definition}

Observe that without requirement $(i)$ the definition would be trivial since the oracle could just
answer $0$ for every vertex, as every vertex can be colored $0$ in some $2$-coloring of $\hyp$. Requirement $(ii)$ is also
necessary, since if the algorithm is allowed to keep $n$ bits of information, then when it is first called it can find a legal $2$-coloring, store it in memory, and then use it to answer subsequent calls.
A moment's reflection reveals\footnote{Indeed, given $u$ the algorithm asks about all possible edges of $\hyp$, then looks for the lexicographically first $2$-coloring of $\hyp$ and returns the color it assigns to $u$.} that there is in fact a (not very efficient) algorithm satisfying Definition~\ref{deforacle}.

Given Theorems \ref{algorithm theorem} and \ref{randomized algorithm} and the success of partition oracles, it is natural to ask
if it is possible to design a coloring oracle that is efficient on average. Let us then introduce the following
definition. If $A$ is a coloring oracle, then we use $T_A(\hyp,u)$ to denote the expected time it takes $A$ to return a color for $u$
in $\hyp$. Since we are considering average case behavior, it might seem natural to define $T_A(\hyp)$ as the average of
$T_A(\hyp,u)$ over all vertices of $\hyp$, but we instead make a stronger requirement and define it as the {\em worst case} over all vertices, that is $T_A(\hyp)=\max_uT_A(\hyp,u)$. We finally define $T_A(n)$ as the average of $T_A(\hyp)$ over all $2$-colorable $k$-graphs on $n$ vertices.
Since hypergraph $2$-colorability is NP-hard we certainly do not expect to have a coloring oracle $A$ for which $T_A(\hyp)$
is sub-exponential for every $\hyp$. Surprisingly, there is a coloring oracle that on average is as efficient as possible.

\begin{theorem} \label{main theorem}
There is a coloring oracle $A$ satisfying $T_{A}(n)=O(1)$. Furthermore, the oracle does not use {\bf any} memory for keeping information between
successive calls.
\end{theorem}

We find it quite surprising that it is possible to not use any shared memory between successive calls, and still answer every query consistently in average case $O(1)$ time. Thus the main conceptual contribution
of this work is demonstrating just how much more efficient can Oracles and LCAs be on average, compared to their worst case behavior.

Comparing coloring oracles and partition oracles, in addition to differing algorithmic goals\footnote{A partition oracle seeks to partition the graph into $O(n)$ sets with as few edges as possible between them, while a coloring oracle's goal is to partition the graph into $O(1)$ sets with no edges inside them.}, note that a coloring oracle handles all inputs and solves the $2$-coloring problem exactly, while a partition oracle only handles planar graphs and only solves an approximate problem.
On the other hand, partition oracles always work in $O(1)$ time, while a coloring oracle only works in $O(1)$ time on average over the set of colorable graphs.

\subsubsection{Transforming the Coloring Oracle into an LCA}\label{subseclca}

In recent years, there has been extensive work on a new model of distributed computing known as Local Computation Algorithms (LCAs) (see Definition~\ref{LCA definition})~\cite{rubinfeld2011fastlocalcomputationalgorithms, alon2011spaceefficientlocalcomputationalgorithms, beyond_worst_case_LCA}.
In this model, the algorithm is given probe access to the input object (in this case, the $k$-graph) and a fixed random string, and must answer queries regarding a particular combinatorial structure defined on it (in this case, the color of a vertex $u$ in a $2$-coloring).
The answers must be globally consistent, each query must be answered with sublinear work, and there is no persistent memory between queries. We observe that our result implies an \textit{average-case} LCA for coloring, which we now discuss.

LCAs for vertex coloring (over worst-case graphs) have been the subject of extensive study. Much of this has focused on $(\Delta +1)$-coloring, where $\Delta$ is the maximum degree of the input graph. First, it is known that $r$-round algorithms in the distributed LOCAL model over graphs with maximum degree $\Delta$ can be used to construct LCAs with query processing time $O(\Delta^r)$ using a reduction of Parnas and Ron \cite{parnas2007approximating}. Applying this to the result of \cite{chang2018optimal} in the distributed LOCAL model yields a $\Delta^{O(\log^*(\Delta))}\cdot \log(n)$ time $(\Delta + 1)$-coloring LCA. Notably, \cite{chang2018complexitydelta1coloring} later constructed a $\Delta^{O(1)}\cdot O(\log(n))$ time LCA for $(\Delta+1)$-coloring, also proving results in other distributed and sublinear models.
For hypergraphs, when a two-coloring is guaranteed by the Lov\'asz Local Lemma, the work of \cite{dorobisz2023local}
gives an LCA that answers queries in polylogarithmic time.  Note that the latter result applies when there
is a bound on the degree of the hypergraph, whereas the graphs we consider are usually dense.
Recent papers have also explored LCAs for 3-coloring and 2-coloring graphs and hypergraphs
typically with additional assumptions made, such as access to linear preprocessing probes \cite{10.1007/978-3-030-04693-4_12}, or answering only up to polylogarithmic many queries \cite{achlioptas2020simple}.

Graph and hypergraph coloring have been studied in other sublinear access models \cite{assadi2019sublinear} and through the lens of property testing \cite{czumaj2001testing, czumaj2005abstract, AaronsonCC25}. Additionally, substantial effort has gone into designing (worst-case) LCAs for a variety of other
fundamental problems, including
maximal independent set~\cite{10.5555/2884435.2884455, ghaffari2022localcomputationmaximalindependent,
M.ghaffari2018sparsifyingdistributedalgorithmsramifications, levi2015localcomputationalgorithmsgraphs},
and maximal matching~\cite{levi2015localcomputationalgorithmsgraphs, 10.1145/1536414.1536447, mansour2013localcomputationapproximationscheme}.

A recent paper \cite{beyond_worst_case_LCA} initiated the study of LCAs over average-case inputs. An oracle is an \textit{average-case LCA} for a problem $\Pi$ over a distribution $\mathcal{G}$ over objects if, with probability at least $(1 - \frac{1}{n})$ over $G \leftarrow \mathcal{G}$, the oracle is an LCA for $\Pi$ on $G$. For hypergraph $2$-coloring, the key distinction between average-case LCAs and coloring oracles is the allowed failure probability (the LCA can fail on some inputs, while the coloring oracle must always return a $2$-coloring).

We observe that our result implies a very efficient average-case LCA for coloring, which also holds for $2$-colorable $2$-graphs due to the different failure criterion.
\begin{theorem}\label{LCA partition oracle}
For all $k \geq 2$, there is an average-case LCA for the uniform distribution over $2$-colorable $k$-graphs with worst-case probe complexity of $O(1)$ and runtime per query of $O(1)$\footnote{This is in the LCA model where the algorithm is given access to a random string consisting of \textit{words} i.e. random entries in $[n]$. A random word corresponds to getting $O(\log n)$ bits of randomness in a step.}. Moreover, the average-case LCA does not require shared randomness between queries.
\end{theorem}
We view this as the most nontrivial example of an average-case LCA so far~\cite{beyond_worst_case_LCA}.
We present the proof in Appendix \ref{LCA partition oracle section}.

\subsection{Key New Idea and Comparison to Prior Works} \label{section:proof-sketch}

The key idea in many average case algorithms is to define a property ${\cal P}$ which is useful in the following sense: on one
hand almost all objects (graphs, hypergraphs, etc) satisfy ${\cal P}$, and on the other hand every object satisfying ${\cal P}$
is easy to solve. As we observed earlier, it is not enough for ${\cal P}$ to hold for $(1-o(1))$-fraction or even for $(1-2^{-n/2})$-fraction
of the objects since that could result in an exponential average case running time if the objects without property ${\cal P}$ take exponential time. It is interesting to note that a similar challenge of coming
up with a useful property appears also in the design of regularity lemmas for graphs \cite{SzemerediReg} and hypergraphs \cite{Gowers,RNSS}.
There, the goal is to come up with a property that is strong enough for the purposes of applying the lemma, and weak enough to be satisfied
by every graph. It is thus no coincidence that the algorithms of Person and Schacht \cite{PersonSchacht1,PersonPolynomialAlog} and Lee, Molla, and Nagle \cite{lee2024twocoloringbipartiteuniformhypergraphs} used useful properties involving notions related to graph and hypergraph regularity.
Hence, they also used the graph/hypergraph regularity lemmas and algorithms \cite{ADLRY}, leading to huge hidden constants and a complicated analysis. We show there is a very simple-to-state useful property, which
we describe in Section \ref{good hypergraph properties}. While it takes
some work to show that most hypergraphs satisfy it, designing an efficient algorithm for hypergraphs satisfying it is almost trivial.
Moreover, the property directly leads to the
coloring oracle of Theorem \ref{main theorem}. Very roughly, the property states that there is a small substructure that has a unique
coloring (what we call $K_{\ell,\ell}$ in Section \ref{good hypergraph properties}) such that the (unique) coloring of this structure
uniquely determines the color of {\em all} vertices of the hypergraph. Most importantly, the coloring of this small substructure forces
the colors of all other vertices via ``paths'' of short length (actually length $2$) and there are in fact ``many'' such paths, making it easy to find one using sampling. We can also take advantage of this property in order to avoid using any memory between successive calls and still maintain the consistency of the coloring. The most (actually, only) technically demanding part of the paper is thus not the design or the analysis of the algorithms, but the proof of Lemma \ref{good hyper graph probability over all bipartite hypergraphs} regarding properties of typical $2$-colorable $k$-graphs. See the end of Section \ref{good hypergraph properties} for a brief description of the proof of this lemma, and Section \ref{section:good-hypergraph-proofs} for the full proof.

\subsection{Paper Overview}

In Section \ref{good hypergraph properties} we formally define the useful hypergraph property that underpins all our algorithms here and state
the key probabilistic fact regarding this property, see Lemma \ref{good hyper graph probability over all bipartite hypergraphs}.
In Section \ref{coloring algorithm section} we prove Theorem \ref{algorithm theorem}. To this end, we present a new average case polynomial time deterministic 2-coloring algorithm which relies on the useful property introduced in Section \ref{good hypergraph properties}.
In Section \ref{oracle section} we prove Theorems \ref{randomized algorithm} and \ref{main theorem}. To do so, we make subtle adjustments to
the algorithm presented in Section \ref{coloring algorithm section}. In Section \ref{section:good-hypergraph-proofs} we prove Lemma \ref{good hyper graph probability over all bipartite hypergraphs} stated in Section \ref{good hypergraph properties}. Finally, in Section \ref{LCA partition oracle section} we introduce an LCA version of the coloring oracle from Section \ref{oracle section}.

\section{A Useful Property for Coloring Hypergraphs}\label{good hypergraph properties}

Our goal in this section is to define the useful hypergraph property alluded to in the proof overview above.
An {\em independent set} $I$ in a $k$-graph is a set of vertices so that none of the edges of the hypergraph is fully contained in $I$.
Note that a $2$-coloring of a hypergraph naturally partitions its vertex set into two independent sets, namely those that are colored red
and those that are colored blue. Given a $2$-colorable $k$-graph $\mathcal{H}$, when we refer to the two independent sets of $\mathcal{H}$, we mean the two independent sets given by some $2$-coloring of $\mathcal{H}$. Note that (unless specified otherwise) we do not assume that $\mathcal{H}$ has a unique $2$-coloring.
For a vertex $u \in \vex$ and a set of vertices $A \subseteq \vex$, we use $N(u,A)$ to denote the set of $(k-1)$-tuples of vertices in $A$ which
form an edge together with $u$. So $N(u,A)$ are the neighbors of vertex $v$ in the set $A$, but as opposed to 2-graphs, where the neighbors of a vertex are also vertices, now the neighbors are $(k-1)$-tuples of vertices.
For an integer $\ell \geq k-1$ the $k$-graph $K^{(k)}_{\ell,\ell}$ is the one consisting
of two vertex sets $A,B$ of size $\ell$ and of all the edges that have a single vertex in either $A$ or $B$ and $k-1$ vertices in the other set.
To simplify the notation, we will use $K_{\ell,\ell}$ instead of $K^{(k)}_{\ell,\ell}$.
We first observe the following simple fact.

\begin{claim}\label{uniquecoloirng}
For every $\ell \geq 2k-3$, the $k$-graph $K_{\ell,\ell}$ has a unique $2$-coloring.
\end{claim}
\begin{proof}
Let $A,B$ be the color classes in the definition of $K_{\ell,\ell}$, and consider any other $2$-coloring. Assume without loss of generality
that two vertices $a,a' \in A$ are assigned different colors. Since\footnote{Note that the assumption $\ell \geq 2k-3$ is indeed needed since when $\ell=2k-4$ there are different $2$-colorings of $K_{\ell,\ell}$ obtained by coloring $k-2$ of the vertices in $A,B$ with color $0/1$.} $\ell \geq 2k-3$ there are $k-1$ vertices in $A$ that received the same color. Suppose this is color $1$, that $a$ is colored $1$, and that $a'$ is colored $0$. It follows from the definition of $K_{\ell,\ell}$ that all vertices in $B$ must be colored $0$, since each $b \in B$ forms an edge with the $k-1$ vertices of $A$ that are colored $1$. But then every edge containing $k-1$ vertices in $B$ and vertex $a'$ is not colored properly.
\end{proof}

Given a copy of $K_{\ell,\ell}$ we will use $A,B$ to denote the two colors classes in its unique $2$-coloring, namely the two sets used
in the definition of $K_{\ell,\ell}$.
We are now ready to define the useful property, which we call {\em good}.

\begin{definition}[Good $k$-Graphs]\label{definition:good-hypergraph}
Suppose $\mathcal{H}$ is an $n$-vertex $2$-colorable $k$-graph. Set
    \begin{equation}\label{l constant function}
    \ell=\ell(k) = 5k.
    \end{equation}
    Then $\mathcal{H}$ is \textit{good} if it satisfies:
    \begin{itemize}
        \item[(i)] The number of copies of $K_{\ell,\ell}$ in $\mathcal{H}$ is at least $n^{2\ell}/2^{2^{10k}}$.
        \item[(ii)] The following holds for every copy $K$ of $K_{\ell,\ell}$ in $\mathcal{H}$. Suppose $A,B$ are the independent sets of $K$.
        Let $N_A$ be the vertices $v$ satisfying $N(v,A) \neq \emptyset$ and $N_B$ be the vertices $v$ satisfying $N(v,B) \neq \emptyset$.
        Then every vertex $u$ in $\mathcal{H}$ satisfies either $|N(u,N_A)| \geq n^{k-1}/k^{4k}$ or $|N(u,N_B)| \geq n^{k-1}/k^{4k}$.
    \end{itemize}
\end{definition}

See Figures 1 and 2 for illustrations of the good property and how the good property assists the algorithms.

Recall that for a property to be useful for our purposes here we first need all but a negligible fraction of the $2$-colorable $k$-graphs
to satisfy it. This is precisely the statement of the following lemma, whose (somewhat technical) proof appears in Section \ref{section:good-hypergraph-proofs}. In what follows $\T^{(k)}_n$ is the set of all $2$-colorable $k$-graphs on $n$ vertices, and
$\hyp\sim\T^{(k)}_n$ means that $\hyp$ is uniformly chosen from~$\T^{(k)}_n$.

\begin{lemma}\label{good hyper graph probability over all bipartite hypergraphs}
If $\hyp\sim\T^{(k)}_n$, then for all large enough $n$, $\hyp$ is good with probability at least $1 - 2^{-2n}$.
\end{lemma}

The main idea of the proof of Lemma \ref{good hyper graph probability over all bipartite hypergraphs} is the following. Consider an alternative probabilistic model for picking a random $2$-colorable $k$-graph on a set $V$ of $n$ vertices, which we denote $\T_{S,n}^{(k)}$: in this case we preselect two ``planted'' independent sets $S$ and $V \setminus S$ (where we will add no edges) and pick each edge intersecting $S$ and $V \setminus S$ with probability $1/2$. It is easy to show (see Lemma \ref{good hypergraph probability theorem}), using classical tail bounds, that a $k$-graph generated by $\T_{S,n}^{(k)}$ such that $n/4 \leq |S| \leq 3n/4$ is good with probability at least $1-2^{-3n}$. What is also ``clear'' is that $\hyp\sim\T_n^{(k)}$ is ``similar'' to $\hyp\sim\mathcal{T}_{S,n}^{(k)}$ when $|S|$ is close to $n/2$. Perhaps surprisingly, there is a very short proof making this intuition precise, thus essentially reducing Lemma \ref{good hyper graph probability over all bipartite hypergraphs}
to Lemma \ref{good hypergraph probability theorem}.

\section{Coloring Bipartite Hypergraphs}\label{coloring algorithm section}

Our goal in this section is to prove Theorem \ref{algorithm theorem}.
Throughout this section we use $\ell=5k$ as in Definition \ref{definition:good-hypergraph}.
By Lemma \ref{good hyper graph probability over all bipartite hypergraphs} only a tiny fraction of the bipartite $k$-graphs are not good. It is thus easy to see that in order to prove Theorem \ref{algorithm theorem} it is enough to design a deterministic algorithm which finds in polynomial time a $2$-coloring of every good $2$-colorable $k$-graphs.
For notational reasons we will sometimes use the term ``bipartite'' instead of ``$2$-colorable''.

Before presenting the concrete algorithm, we give an informal description of it.
The algorithm first uses exhaustive search in order to find a copy of $K_{\ell,\ell}$.
If $\hyp$ is good then we know that it has many copies of $K_{\ell,\ell}$.
Letting $A,B$ be the independent sets in the unique $2$-coloring the $K_{\ell,\ell}$ the algorithm found, it then colors $A$ with $0$ and $B$ with $1$.
The algorithm now looks for all vertices $v$ which form an edge with $k-1$ vertices from $A$ or $B$.
Note that the color of such a $v$ is uniquely determined by the coloring of $A$ and $B$ and thus we color them accordingly.
The algorithm now looks for all vertices $v$ which form an edge with $k-1$ vertices that were colored $0$ or $1$ in the previous step.
Again, the color of such a $v$ is uniquely determined by the coloring we made in the previous step.
If $\hyp$ is good we thus color all of its vertices.
If any step of the algorithm fails, it just uses exhaustive search to find a legal $2$-coloring.
We remind the reader that $N(u, A)$ denotes the set of $(k-1)$-tuples of vertices in $A$ which form an edge with $u$.

\input{fig1}

\begin{algorithm}[H]
\caption{Deterministic Coloring Algorithm for Bipartite $k$-Graphs}
    \textbf{Input:} A bipartite $k$-graph $\mathcal{H}$.\\
    \textbf{Output:} A proper $2$-coloring of $\mathcal{H}$.\\

    \textbf{Procedure:} \textsc{Color-Bipartite-Hypergraph($\mathcal{H}$)}\\
    1. \quad Search for a copy of $K_{\ell,\ell}$ using exhaustive search.\\
    2. \quad \textbf{If} no copy of $K_{\ell,\ell}$ was found, proceed to Step 13.\\
    3. \quad Let $A,B$ be the independent sets of the copy of $K_{\ell,\ell}$ found in Step $1$. \\
    4. \quad Color the vertices in $A$ with color $0$ and those in $B$ with color $1$.\\
    5. \quad $\forall u\in V(H)$ \textbf{do}\\
    6. \qquad \textbf{If} $N(u, B)\neq \emptyset$, then color $u$ with color $0$.\\
    7. \qquad \textbf{If} $N(u, A)\neq \emptyset$, then color $u$ with color $1$.\\
    8. \quad Let $C_0$ (resp. $C_1$) be the vertices colored $0$ (resp. $1$) thus far.\\
    9. \quad $\forall u\in V(H)\setminus \{C_0 \cup C_1\}$ \textbf{do}\\
    10. \qquad \textbf{If} $N(v,C_1)\neq \emptyset$, then color $u$ with color $0$.\\
    11. \qquad \textbf{If} $N(v,C_0)\neq \emptyset$, then color $u$ with color $1$.\\
    12.\quad \textbf{If} all vertices were colored then return this $2$-coloring, \textbf{Else} go to Step 13.\\
    13.\quad Exhaustively search for a legal $2$-coloring of $\mathcal{H}$. Return the first one found.
\end{algorithm}

We begin by demonstrating that \textsc{Color-Bipartite-Hypergraph} indeed produces a proper $2$-coloring.

\begin{lemma}\label{proper coloring algorithm} Suppose $\mathcal{H}$ is a bipartite $k$-graph. Then the algorithm \textsc{Color-Bipartite-Hypergraph} returns a proper $2$-coloring of $\mathcal{H}$.
\end{lemma}

\begin{proof}
If a copy of $K_{\ell,\ell}$ is not found, then we run Step 13 which finds a legal $2$-coloring (since one exists) so
the claim holds in this case. Suppose then that a copy of $K_{\ell,\ell}$ was found, and fix a legal $2$-coloring $c:\vex \mapsto \{0,1\}$ of the vertices of $\mathcal{H}$.
Recall that by Claim \ref{uniquecoloirng} $K_{\ell,\ell}$ has unique $2$-coloring. Hence, $c$ assigns all vertices of $A$ (resp. $B$) the same color. Assume without loss of generality that these are colors $0$ and $1$ as in our coloring.
Clearly, $c$ colors the vertices colored in Steps 6/7 in the same color as the algorithm does.
Similarly, $c$ colors the vertices colored in Steps 10/11 in the same color as the algorithm does.
Hence, if we colored all of the vertices of $\mathcal{H}$ then our coloring agrees with $c$ which is a legal $2$-coloring of $\mathcal{H}$ (which means that $c$ is the unique $2$-coloring of $\mathcal{H}$).
If the algorithm did not color all the vertices then it again resorts to exhaustively looking for a legal $2$-coloring.
\end{proof}

For a given $2$-coloring algorithm $A$ and input $\mathcal{H}$, let $T_A(\mathcal{H})$ denote the running time of $A$ on input $\mathcal{H}$.
Let $T_A(n)$ be the average of $T_A(\mathcal{H})$ over all $n$-vertex bipartite $k$-graphs $\mathcal{H}$, that is, the average case running time of $A$ over the $n$-vertex bipartite $k$-graphs.

\begin{lemma}\label{algorithm time complexity lemma}
The algorithm \textsc{Color-Bipartite-Hypergraph} satisfies
$$
T_\textsc{Color-Bipartite-Hypergraph}(n) = n^{O(k)}\;.
$$
\end{lemma}

\begin{proof}
    We begin by analyzing the running time of each step. First, for Step 1, a copy of $K_{\ell,\ell}$ can be found (if it exists) by enumerating over all size $2\ell$ subsets of the $n$ vertices. Since $\ell$ is $O(k)$, a copy of $K_{\ell,\ell}$ can be found in time $n^{O(k)}$.
    In Steps 6/7, for each vertex $u \in V(\mathcal{H})$, we check whether $N(u, A)\neq \emptyset$ or $N(u, B)\neq \emptyset$. This requires $2\binom{\ell}{k-1}=O(1)$ queries per vertex, leading to a total running time of $O(n)$.
    In steps 10-11, for each vertex $u \in V(\mathcal{H})$, we check whether $N(v,C_1)\neq \emptyset$ or $N(v,C_2)\neq \emptyset$. This requires $2\binom{n}{k-1} = O(n^{k-1})$ queries per vertex, resulting in total running time of $O(n^k)$.
    If Step 13 is reached then the running time is $O(\poly(n)\cdot 2^n)$.

    The crucial observation now is that if \(\hyp\) is good then the algorithm will find a legal $2$-coloring of \(\hyp\) without resorting to the exhaustive search in Step 13. Indeed, the algorithm will find a copy of $K_{\ell,\ell}$ (by property $(i)$ of good $k$-graphs). It will then color in Steps 6/7 some vertices with color $0$ and some with color $1$ (the sets $C_0,C_1$ in the algorithm which correspond to the sets $N_A,N_B$ in the definition of good $k$-graphs). It will then color all the remaining vertices in Steps 10-11 (by Property $(ii)$ of good $k$-graphs). Hence, if \(\hyp\) is good then the algorithm runs in
    time $n^{O(k)}$. By Lemma \ref{good hyper graph probability over all bipartite hypergraphs}, the probability that $\mathcal{H} \sim \T^{(k)}_{n}$ is not good is at most $2^{-2n}$. Hence, we conclude that the average running time of the algorithm over all bipartite $k$-graphs satisfies
    \[
        T_\textsc{Color-Bipartite-Hypergraph}(n) = n^{O(k)} + \poly(n)\cdot 2^n\cdot 2^{-2n}  = n^{O(k)}. \qedhere\]
\end{proof}

\begin{proof}[Proof of Theorem \ref{algorithm theorem}]
By Lemma \ref{proper coloring algorithm} \textsc{Color-Bipartite-Hypergraph} always returns a legal $2$-coloring, and by
Lemma \ref{algorithm time complexity lemma} its average running time is $n^{O(k)}$. It is also clear that the algorithm does not use
any memory to keep information between successive calls to it.
\end{proof}

\section{A Coloring Oracle for Bipartite Hypergraphs}\label{oracle section}
In this section we prove Theorems \ref{main theorem} and \ref{randomized algorithm}.
Throughout this section we use $\ell=5k$ as in Definition \ref{definition:good-hypergraph}.
To this end, we will modify the algorithm \textsc{Color-Bipartite-Hypergraph} presented in the previous section in order to turn it into a coloring oracle and subsequently into a randomized algorithm with linear expected running time.
Let us describe the key ideas needed to make the running time $O(1)$ on average and how to avoid using any memory between successive calls
and still maintain a consistent coloring. Observe that in order to get running time $O(1)$ on average it is enough to obtain
this running time for good hypergraphs. For such inputs, we can in fact find a copy of $K_{\ell,\ell}$ in $O(1)$ time
since a positive proportion of all $2\ell$-tuples contain a $K_{\ell,\ell}$ (see Step 2 of Algorithm 2 below). It is easy to see that if $\mathcal{H}$ is good then a coloring of a copy of $K_{\ell,\ell}$ forces a coloring of every vertex $u$ in $\mathcal{H}$. In fact, there are many
$(k-1)$-tuples $v_1,\ldots,v_{k-1}$ which witness this fact, so the color of $u$ can be deduced in $O(1)$ time (as in Steps 4-5).
The challenge is that the Oracle should answer consistently for {\em every} input, hence we need a mechanism
for solving this issue. This is achieved by Step 2, in which we not only try to find a copy of $K_{\ell,\ell}$ but we also demand
that this copy forces a color for vertex $1$. We always assume that vertex $1$ is colored $0$ (see below why), so in this sense vertex
$1$ forces the colors of $A$ and $B$. The colors of $A$ and $B$ then force the colors of other vertices $u$ (but {\em not} necessarily all of them) in Steps 4-5. If we think of a $2$-coloring of $\mathcal{H}$ as a $0/1$ string of length $n$, then the color of vertex $1$ is the most significant bit, hence the lexicographically first legal $2$-coloring of $\mathcal{H}$ must assign vertex $1$ the color $0$ (since flipping the colors of a legal coloring is also a legal coloring). This is why it is convenient to also look for a coloring giving $1$ the color $0$.

\begin{algorithm}
\caption{Coloring Oracle for Bipartite $k$-Graphs}
    \textbf{Input:} Vertex $u$ in some bipartite $k$-graph $\mathcal{H}$, and oracle access to $E(\mathcal{H})$. \\
    \textbf{Output:} Color assignment to $u$.\\

    1.\quad\textbf {Procedure:} \textsc{Coloring-Oracle($\mathcal{H}$,$u$)}\\
    2.\quad Uniformly sample $(2\ell+k-1)$-tuples $(x_1,\ldots,x_{2\ell},y_1,\ldots,y_{k-1})$ of vertices until $(i)$ and $(ii)$ hold:\\
    $~~~~~~~$ $(i)$ Vertices $x_1,\ldots,x_{2\ell}$ span a copy of $K_{\ell,\ell}$. Set $A,B$ to be the independent sets of this copy.\\
    $~~~~~~~$ $(ii)$ Vertices $y_1,\ldots,y_{k-1}$ satisfy one of the following:\\
    $~~~~~~~~$ $(ii.1):$ $(1,y_1,\dots,y_{k-1})\in E(\mathcal{H})$ and $N(y_i,A)\neq \emptyset$ for every $i\in [k-1]$. Set $c_A=0,c_B=1$\\
    $~~~~~~~~$ $(ii.2):$ $(1,y_1,\dots,y_{k-1})\in E(\mathcal{H})$ and $N(y_i,B)\neq \emptyset$ for every $i\in [k-1]$. Set $c_A=1,c_B=0$\\
    $~~~~~~~$ \textbf{If} all $(2\ell+k-1)$-tuples were inspected, and none of them satisfies $(i),(ii)$, goto Step 7.\\
    3.\quad Uniformly sample $(k-1)$-tuples of vertices $v_1,\dots,v_{k-1}$:\\
    4.\qquad \textbf{If} $(u,v_1,\dots,v_{k-1})\in E(\mathcal{H})$ and $N(v_i,A)\neq \emptyset$ for every $i\in [k-1]$ then return $c_A$.\\
    5.\qquad \textbf{If} $(u,v_1,\dots,v_{k-1})\in E(\mathcal{H})$ and $N(v_i,B)\neq \emptyset$ for every $i\in [k-1]$ then return $c_B$.\\
    6.\qquad \textbf{If} all $(k-1)$-tuples have been inspected, and none of them satisfies $4-5$, then goto Step 7.\\
    7.\quad Exhaustively search for the lexicographically first legal $2$-coloring $c$ of $\mathcal{H}$. Return $c(v)$.
\end{algorithm}

\begin{lemma}\label{oracle corectness}
Algorithm \textsc{Coloring-Oracle} is a $2$-coloring oracle. That is,
if $\mathcal{H}$ is a bipartite $k$-graph, then for every sequence of queries $u_1,u_2,\ldots$ the oracle's answers are consistent
with the same legal $2$-coloring of $\mathcal{H}$.
\end{lemma}

\begin{proof} Let $c$ be the lexicographically first legal $2$-coloring of $\mathcal{H}$. Note that such a coloring must
assign vertex $1$ the color $0$. We claim that for any sequence of queries,
the oracle's answers are always consistent with $c$. This is certainly the case if when applied to vertex $u$, the algorithm ever resorts to Step 7, so suppose it does not. In this case we know that the algorithm found a copy of $K_{\ell,\ell}$, a $(k-1)$-tuple of vertices satisfying
conditions $(ii.1)$ or $(ii.2)$ and a $(k-1)$-tuple satisfying either Step 4 or 5. Suppose wlog that steps $(ii.1)$ and 4 are the ones that were satisfied (the other 3 options are identical). First recall that by Claim \ref{uniquecoloirng} every legal $2$-coloring of $\mathcal{H}$ gives
the vertices of $A$ the same color and to those in $B$ the other color. If $(ii.1)$ holds then each of the vertices $y_1,\ldots,y_{k-1}$ forms an edge with a $(k-1)$-tuple in $A$ implying that in any legal $2$-coloring of $\mathcal{H}$ the vertices $y_1,\ldots,y_{k-1}$ are colored as $B$. Since $(ii.1)$ further assumes that $(1,y_1,\dots,y_{k-1}) \in E(\mathcal{H})$ we conclude that in any legal $2$-coloring of $\mathcal{H}$ the set $A$ is assigned the same color as vertex $1$. In particular, this means that $c$ assigns $A$ the color
$0$ and $B$ the color $1$. Hence we set $c_A=0,c_B=1$ to indicate this fact. By an identical argument, if Step 4 holds then in any legal $2$-coloring of $E(\mathcal{H})$, and in particular in $c$, vertex $u$ receives the same color as $A$. Hence returning color $c_A$ for $u$ is consistent with $c$.
\end{proof}

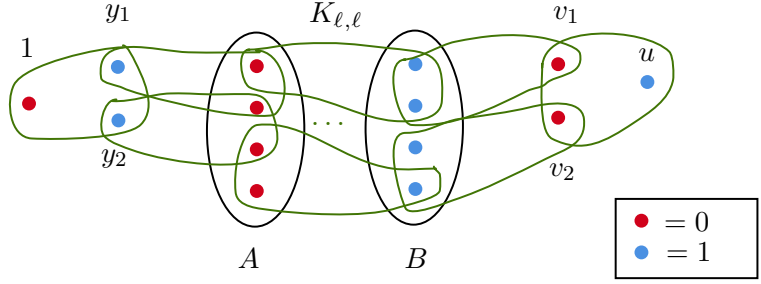
\begin{figure}\label{fig:fig-2}
    \centering

\tikzset{every picture/.style={line width=0.75pt}} 

\begin{tikzpicture}[x=0.75pt,y=0.75pt,yscale=-1,xscale=1]

\draw   (323.87,57.04) .. controls (337.43,57.17) and (348.22,79.19) .. (347.97,106.25) .. controls (347.73,133.3) and (336.54,155.13) .. (322.98,155) .. controls (309.42,154.88) and (298.63,132.85) .. (298.87,105.8) .. controls (299.12,78.75) and (310.31,56.92) .. (323.87,57.04) -- cycle ;
\draw  [color={rgb, 255:red, 208; green, 2; blue, 27 }  ,draw opacity=1 ][fill={rgb, 255:red, 208; green, 2; blue, 27 }  ,fill opacity=1 ] (321,74) .. controls (321,72.34) and (322.34,71) .. (324,71) .. controls (325.66,71) and (327,72.34) .. (327,74) .. controls (327,75.66) and (325.66,77) .. (324,77) .. controls (322.34,77) and (321,75.66) .. (321,74) -- cycle ;
\draw  [color={rgb, 255:red, 208; green, 2; blue, 27 }  ,draw opacity=1 ][fill={rgb, 255:red, 208; green, 2; blue, 27 }  ,fill opacity=1 ] (321,95) .. controls (321,93.34) and (322.34,92) .. (324,92) .. controls (325.66,92) and (327,93.34) .. (327,95) .. controls (327,96.66) and (325.66,98) .. (324,98) .. controls (322.34,98) and (321,96.66) .. (321,95) -- cycle ;
\draw  [color={rgb, 255:red, 208; green, 2; blue, 27 }  ,draw opacity=1 ][fill={rgb, 255:red, 208; green, 2; blue, 27 }  ,fill opacity=1 ] (321,116) .. controls (321,114.34) and (322.34,113) .. (324,113) .. controls (325.66,113) and (327,114.34) .. (327,116) .. controls (327,117.66) and (325.66,119) .. (324,119) .. controls (322.34,119) and (321,117.66) .. (321,116) -- cycle ;
\draw  [color={rgb, 255:red, 208; green, 2; blue, 27 }  ,draw opacity=1 ][fill={rgb, 255:red, 208; green, 2; blue, 27 }  ,fill opacity=1 ] (321,137) .. controls (321,135.34) and (322.34,134) .. (324,134) .. controls (325.66,134) and (327,135.34) .. (327,137) .. controls (327,138.66) and (325.66,140) .. (324,140) .. controls (322.34,140) and (321,138.66) .. (321,137) -- cycle ;
\draw   (403.87,56.04) .. controls (417.43,56.17) and (428.22,78.19) .. (427.97,105.25) .. controls (427.73,132.3) and (416.54,154.13) .. (402.98,154) .. controls (389.42,153.88) and (378.63,131.85) .. (378.87,104.8) .. controls (379.12,77.75) and (390.31,55.92) .. (403.87,56.04) -- cycle ;
\draw  [color={rgb, 255:red, 74; green, 144; blue, 226 }  ,draw opacity=1 ][fill={rgb, 255:red, 74; green, 144; blue, 226 }  ,fill opacity=1 ] (401,73) .. controls (401,71.34) and (402.34,70) .. (404,70) .. controls (405.66,70) and (407,71.34) .. (407,73) .. controls (407,74.66) and (405.66,76) .. (404,76) .. controls (402.34,76) and (401,74.66) .. (401,73) -- cycle ;
\draw  [color={rgb, 255:red, 74; green, 144; blue, 226 }  ,draw opacity=1 ][fill={rgb, 255:red, 74; green, 144; blue, 226 }  ,fill opacity=1 ] (401,94) .. controls (401,92.34) and (402.34,91) .. (404,91) .. controls (405.66,91) and (407,92.34) .. (407,94) .. controls (407,95.66) and (405.66,97) .. (404,97) .. controls (402.34,97) and (401,95.66) .. (401,94) -- cycle ;
\draw  [color={rgb, 255:red, 74; green, 144; blue, 226 }  ,draw opacity=1 ][fill={rgb, 255:red, 74; green, 144; blue, 226 }  ,fill opacity=1 ] (401,115) .. controls (401,113.34) and (402.34,112) .. (404,112) .. controls (405.66,112) and (407,113.34) .. (407,115) .. controls (407,116.66) and (405.66,118) .. (404,118) .. controls (402.34,118) and (401,116.66) .. (401,115) -- cycle ;
\draw  [color={rgb, 255:red, 74; green, 144; blue, 226 }  ,draw opacity=1 ][fill={rgb, 255:red, 74; green, 144; blue, 226 }  ,fill opacity=1 ] (401,136) .. controls (401,134.34) and (402.34,133) .. (404,133) .. controls (405.66,133) and (407,134.34) .. (407,136) .. controls (407,137.66) and (405.66,139) .. (404,139) .. controls (402.34,139) and (401,137.66) .. (401,136) -- cycle ;
\draw  [color={rgb, 255:red, 65; green, 117; blue, 5 }  ,draw opacity=1 ][line width=0.75] [line join = round][line cap = round] (319.33,65.67) .. controls (328.16,65.67) and (336.49,63.09) .. (345.33,62.67) .. controls (366.14,61.68) and (372.53,62.91) .. (388.33,64.67) .. controls (410.24,67.1) and (419.09,65.3) .. (417.33,91.67) .. controls (416.68,101.54) and (402.04,102.34) .. (393.33,101.67) .. controls (377.17,100.42) and (361.98,90.9) .. (346.33,88.67) .. controls (337.88,87.46) and (326.95,87.98) .. (320.33,84.67) .. controls (316.24,82.62) and (313.28,65.67) .. (320.33,65.67) ;
\draw  [color={rgb, 255:red, 65; green, 117; blue, 5 }  ,draw opacity=1 ][line width=0.75] [line join = round][line cap = round] (414.33,125.67) .. controls (407,125.67) and (399.59,126.7) .. (392.33,125.67) .. controls (366.49,121.98) and (348.38,100.99) .. (324.33,103.67) .. controls (319.15,104.24) and (315.32,115.75) .. (314.33,120.67) .. controls (313.22,126.23) and (312.74,132.02) .. (313.33,137.67) .. controls (314.15,145.45) and (325.69,145.86) .. (331.33,146.67) .. controls (342.97,148.33) and (378.37,150.99) .. (394.33,145.67) .. controls (398.42,144.31) and (415.84,140.65) .. (416.33,136.67) .. controls (416.66,134.02) and (416.66,131.31) .. (416.33,128.67) .. controls (416.2,127.61) and (413.33,126.05) .. (413.33,124.67) ;
\draw  [color={rgb, 255:red, 208; green, 2; blue, 27 }  ,draw opacity=1 ][fill={rgb, 255:red, 208; green, 2; blue, 27 }  ,fill opacity=1 ] (473,73) .. controls (473,71.34) and (474.34,70) .. (476,70) .. controls (477.66,70) and (479,71.34) .. (479,73) .. controls (479,74.66) and (477.66,76) .. (476,76) .. controls (474.34,76) and (473,74.66) .. (473,73) -- cycle ;
\draw  [color={rgb, 255:red, 208; green, 2; blue, 27 }  ,draw opacity=1 ][fill={rgb, 255:red, 208; green, 2; blue, 27 }  ,fill opacity=1 ] (473,100) .. controls (473,98.34) and (474.34,97) .. (476,97) .. controls (477.66,97) and (479,98.34) .. (479,100) .. controls (479,101.66) and (477.66,103) .. (476,103) .. controls (474.34,103) and (473,101.66) .. (473,100) -- cycle ;
\draw  [color={rgb, 255:red, 74; green, 144; blue, 226 }  ,draw opacity=1 ][fill={rgb, 255:red, 74; green, 144; blue, 226 }  ,fill opacity=1 ] (518,82) .. controls (518,80.34) and (519.34,79) .. (521,79) .. controls (522.66,79) and (524,80.34) .. (524,82) .. controls (524,83.66) and (522.66,85) .. (521,85) .. controls (519.34,85) and (518,83.66) .. (518,82) -- cycle ;
\draw  [color={rgb, 255:red, 65; green, 117; blue, 5 }  ,draw opacity=1 ][line width=0.75] [line join = round][line cap = round] (398,70.67) .. controls (405.41,70.67) and (415.93,63.28) .. (424,61.67) .. controls (438.19,58.83) and (451.43,56.91) .. (468,59.67) .. controls (470.87,60.15) and (486.19,63.62) .. (487,67.67) .. controls (489.33,79.34) and (474.57,78.09) .. (470,82.67) .. controls (467.49,85.18) and (463.99,84.67) .. (460,86.67) .. controls (443.58,94.88) and (423.5,105.1) .. (402,103.67) .. controls (395.26,103.22) and (386.92,70.67) .. (399,70.67) ;
\draw  [color={rgb, 255:red, 65; green, 117; blue, 5 }  ,draw opacity=1 ][line width=0.75] [line join = round][line cap = round] (399,107.67) .. controls (411.66,107.67) and (416.64,102) .. (426,100.67) .. controls (438.63,98.86) and (451.35,97.33) .. (462,94.67) .. controls (469.54,92.78) and (487.92,89.86) .. (489,100.67) .. controls (490.28,113.49) and (479.65,115.84) .. (470,120.67) .. controls (453.97,128.68) and (439.79,136.07) .. (423,141.67) .. controls (419.93,142.69) and (401.44,150.54) .. (399,145.67) .. controls (395.83,139.33) and (387.23,107.67) .. (399,107.67) -- cycle ;
\draw  [color={rgb, 255:red, 65; green, 117; blue, 5 }  ,draw opacity=1 ][line width=0.75] [line join = round][line cap = round] (468,74.67) .. controls (468,52.36) and (519.55,54.21) .. (532,66.67) .. controls (534.32,68.99) and (534.28,78.97) .. (534,80.67) .. controls (532.28,91.01) and (521.85,97.43) .. (514,102.67) .. controls (499.73,112.18) and (488.48,118.41) .. (473,110.67) .. controls (465.46,106.9) and (468,79.97) .. (468,72.67) ;
\draw  [color={rgb, 255:red, 74; green, 144; blue, 226 }  ,draw opacity=1 ][fill={rgb, 255:red, 74; green, 144; blue, 226 }  ,fill opacity=1 ] (257.22,101.35) .. controls (257.24,103.01) and (255.91,104.37) .. (254.25,104.38) .. controls (252.6,104.4) and (251.24,103.07) .. (251.22,101.42) .. controls (251.2,99.76) and (252.53,98.4) .. (254.19,98.38) .. controls (255.85,98.37) and (257.2,99.69) .. (257.22,101.35) -- cycle ;
\draw  [color={rgb, 255:red, 74; green, 144; blue, 226 }  ,draw opacity=1 ][fill={rgb, 255:red, 74; green, 144; blue, 226 }  ,fill opacity=1 ] (256.93,74.35) .. controls (256.95,76.01) and (255.62,77.37) .. (253.96,77.38) .. controls (252.3,77.4) and (250.95,76.07) .. (250.93,74.42) .. controls (250.91,72.76) and (252.24,71.4) .. (253.9,71.39) .. controls (255.55,71.37) and (256.91,72.7) .. (256.93,74.35) -- cycle ;
\draw  [color={rgb, 255:red, 208; green, 2; blue, 27 }  ,draw opacity=1 ][fill={rgb, 255:red, 208; green, 2; blue, 27 }  ,fill opacity=1 ] (212.13,92.84) .. controls (212.14,94.5) and (210.82,95.85) .. (209.16,95.87) .. controls (207.5,95.89) and (206.14,94.56) .. (206.13,92.9) .. controls (206.11,91.25) and (207.44,89.89) .. (209.09,89.87) .. controls (210.75,89.85) and (212.11,91.18) .. (212.13,92.84) -- cycle ;
\draw  [color={rgb, 255:red, 65; green, 117; blue, 5 }  ,draw opacity=1 ][line width=0.75] [line join = round][line cap = round] (325.5,67) .. controls (313.5,67) and (301.5,67.32) .. (289.5,67) .. controls (273.26,66.56) and (242.06,55.94) .. (245.5,80) .. controls (245.88,82.67) and (253.56,83.53) .. (256.5,85) .. controls (263.82,88.66) and (301.03,97.13) .. (311.5,98) .. controls (314.44,98.24) and (331.11,100.39) .. (333.5,98) .. controls (342.85,88.65) and (336.95,66) .. (323.5,66) ;
\draw  [color={rgb, 255:red, 65; green, 117; blue, 5 }  ,draw opacity=1 ][line width=0.75] [line join = round][line cap = round] (256.5,63.5) .. controls (258.83,63.5) and (261.7,69.89) .. (262.5,71.5) .. controls (266.05,78.59) and (274.37,94.63) .. (265.5,103.5) .. controls (258.29,110.71) and (223.88,111.2) .. (210.5,110.5) .. controls (199.29,109.91) and (197.36,91.79) .. (201.5,83.5) .. controls (206.83,72.84) and (245.08,64.5) .. (256.5,64.5) ;
\draw  [color={rgb, 255:red, 65; green, 117; blue, 5 }  ,draw opacity=1 ][line width=0.75] [line join = round][line cap = round] (251.16,91.62) .. controls (261.7,92.19) and (270.73,87.98) .. (280.41,87.77) .. controls (292.19,87.52) and (293.01,88.35) .. (305.42,88.22) .. controls (326.04,88) and (328.41,86.33) .. (332.8,100.53) .. controls (333.55,102.95) and (334.96,108.66) .. (334.17,112.34) .. controls (330.21,130.85) and (278.37,122.54) .. (261.84,116.58) .. controls (256.48,114.65) and (250.72,113.01) .. (247.2,109.47) .. controls (243.37,105.63) and (246.58,90.47) .. (253.21,90.83) ;
\draw  [color={rgb, 255:red, 208; green, 2; blue, 27 }  ,draw opacity=1 ][fill={rgb, 255:red, 208; green, 2; blue, 27 }  ,fill opacity=1 ] (515,152) .. controls (515,150.34) and (516.34,149) .. (518,149) .. controls (519.66,149) and (521,150.34) .. (521,152) .. controls (521,153.66) and (519.66,155) .. (518,155) .. controls (516.34,155) and (515,153.66) .. (515,152) -- cycle ;
\draw  [color={rgb, 255:red, 74; green, 144; blue, 226 }  ,draw opacity=1 ][fill={rgb, 255:red, 74; green, 144; blue, 226 }  ,fill opacity=1 ] (515,168) .. controls (515,166.34) and (516.34,165) .. (518,165) .. controls (519.66,165) and (521,166.34) .. (521,168) .. controls (521,169.66) and (519.66,171) .. (518,171) .. controls (516.34,171) and (515,169.66) .. (515,168) -- cycle ;
\draw   (505,141) -- (576.5,141) -- (576.5,181) -- (505,181) -- cycle ;

\draw (313,164.4) node [anchor=north west][inner sep=0.75pt]    {$A$};
\draw (397,164.4) node [anchor=north west][inner sep=0.75pt]    {$B$};
\draw (349,40.4) node [anchor=north west][inner sep=0.75pt]    {$K_{\ell }{}_{,}{}_{\ell }$};
\draw (349.97,100.65) node [anchor=north west][inner sep=0.75pt]  [color={rgb, 255:red, 65; green, 117; blue, 5 }  ,opacity=1 ]  {$\dotsc $};
\draw (472,42.4) node [anchor=north west][inner sep=0.75pt]    {$v_{1}$};
\draw (470,121.4) node [anchor=north west][inner sep=0.75pt]    {$v_{2}$};
\draw (515,63.4) node [anchor=north west][inner sep=0.75pt]    {$u$};
\draw (246,41.4) node [anchor=north west][inner sep=0.75pt]    {$y_{1}$};
\draw (244,114.4) node [anchor=north west][inner sep=0.75pt]    {$y_{2}$};
\draw (203,58.4) node [anchor=north west][inner sep=0.75pt]    {$1$};
\draw (528,146.4) node [anchor=north west][inner sep=0.75pt]    {$=0$};
\draw (529,161.4) node [anchor=north west][inner sep=0.75pt]    {$=1$};

\end{tikzpicture}
    \caption{Illustration of how vertices are colored in \textit{good} $3$-graphs (Definition \ref{definition:good-hypergraph}) by Algorithm 2. First, once a copy of $K_{\ell, \ell}$ is found, we must determine its coloring relative to how vertex 1 is colored. Therefore, a path of length up to two to vertex 1 is found, and $K_{\ell, \ell}$ is colored to be consistent with vertex 1 being colored 0 (red). All other vertices will be within a path length of two of this copy and are colored consistently relative to the $K_{\ell, \ell}$ copy.}
    \label{fig:enter-label}
\end{figure}

Recall that if $A$ is a coloring oracle, $\mathcal{H}$ is a bipartite $k$-graph and $u \in V(\mathcal{H})$, then we use
$T_A(\mathcal{H}, u)$ to denote the expected running time it takes $A$ to return a color for $u$. We further set
$T_A(\mathcal{H})=\max_{u\in V(\mathcal{H})} T_A(\mathcal{H},u)$, and $T_A(n)$ as the average of $T_A(\mathcal{H})$
over all bipartite $n$-vertex $k$-graphs $\mathcal{H}$.

\begin{lemma}\label{good hypergraph average case complexity}
The algorithm \textsc{Coloring-Oracle} satisfies
$$
T_\textsc{Coloring-Oracle}(n) = O(1)\;.
$$
\end{lemma}

\begin{proof} We first claim that if $\hyp$ is a good bipartite $k$-graph then
\begin{displaymath}
        T_\textsc{Coloring-Oracle}(\mathcal{H}) = O(1)\;.
\end{displaymath}
Assuming $\hyp$ is a good bipartite $k$-graph, we need to prove that $T_\textsc{Coloring-Oracle}(\mathcal{H},u) = O(1)$ for every $u\in\vex$.
Item $(i)$ in Definition \ref{definition:good-hypergraph} guarantees that if $\hyp$ is good then in Step 2 the random $2\ell$-tuple $x_1,\ldots,x_{2\ell}$ contains a copy of $K_{\ell,\ell}$ with probability at least $2^{-2^{10k}}$. For each such copy let $C_B$ denote the vertices $v$ satisfying $N(v,A)\neq \emptyset$ and let $C_A$ denote the vertices $v$ satisfying
$N(v,B)\neq \emptyset$.
By item $(ii)$ in Definition \ref{definition:good-hypergraph} every vertex $u$ forms an edge with at least $n^{k-1}/k^{2k}$ of the $(k-1)$-tuples
in either $C_A$ or $C_B$. This means that if the $2\ell$ vertices $x_1,\ldots,x_{2\ell}$ form a copy of $K_{\ell,\ell}$, then the probability that
$y_1,\ldots,y_{k-1}$ satisfy $(ii.1)$ or $(ii.2)$ is at least $k^{-4k}$. Since we can think of picking a $(2\ell+k-1)$-tuple as first picking a $2\ell$-tuple and then a $(k-1)$-tuple we conclude that the probability that the $(2\ell+k-1)$-tuple satisfies $(i)$ and either
$(ii.1)$ or $(ii.2)$ is at least $2^{-2^{10k}}\cdot k^{-4k}$. Hence, the expected number of iterations until we find a $(2\ell+k-1)$-tuple satisfying these condition is $O(1)$. After we find such a $(2\ell+k-1)$-tuple, then by the same argument, a random $(k-1)$-tuple satisfies Steps 4 or 5 is at least $k^{-4k}$; therefore this step also takes $O(1)$ in expectation.

We see that if $\mathcal{H}$ is good then by the previous paragraph $T_\textsc{Coloring-Oracle}(\mathcal{H}) = O(1)$, and if $\mathcal{H}$ is not good then $T_\textsc{Coloring-Oracle}(\mathcal{H}) \leq 2^n$. By Lemma \ref{good hyper graph probability over all bipartite hypergraphs} only a $2^{-2n}$-fraction of the bipartite $k$-graphs are not good. Therefore
\[
        T_\textsc{Coloring-Oracle}(n) = O(1) + 2^n\cdot 2^{-2n} = O(1).\qedhere
\]
\end{proof}

\begin{proof}[Proof (of Theorem \ref{main theorem})]
By Lemma \ref{oracle corectness} the algorithm \textsc{Coloring-Oracle} is indeed a coloring oracle, and by Lemma
\ref{good hypergraph average case complexity} it satisfies $T_\textsc{Coloring-Oracle}(n) = O(1)$. Finally, note that the random string used on query $u$ is local and does not need to be maintained between queries, so the algorithm requires no persistent storage (and the responses remain consistent by Lemma \ref{oracle corectness}).
\end{proof}

\begin{proof}[Proof (of Theorem \ref{randomized algorithm})]
Suppose we color $\hyp$ by invoking \textsc{Coloring-Oracle} on all vertices of $\hyp$.
By linearity of expectation, for a given $\hyp$ the expected running time of the algorithm is at most $n\cdot T_\textsc{Coloring-Oracle}(\mathcal{H})$.
By Theorem \ref{main theorem}, this means that the average running time of the algorithm over $\mathcal{H} \sim \T^{(k)}_{n}$ is at most
\[
n \cdot T_\textsc{Coloring-Oracle}(n)= n \cdot O(1)=O(n)\;.  \qedhere
\]
\end{proof}

\begin{remark}
   The coloring oracle algorithm immediately yields an (average-case) LCA (see Definition \ref{def:avg-case-lca}). In Appendix \ref{LCA partition oracle section} we show how to obtain a slightly simpler average-case LCA by modifying the coloring oracle.
\end{remark}

\section{Proof of Lemma \ref{good hyper graph probability over all bipartite hypergraphs}}\label{section:good-hypergraph-proofs}

Recall that $\T_n^{(k)}$ denotes the set of all $2$-colorable graphs on a set $V$ of $n$ vertices.
Instead of directly considering a $k$-graph chosen from $\T_n^{(k)}$ we will instead consider $k$-graph from a family which is significantly easier to analyze and is defined next.

\begin{definition}\label{random hypergraph set definition}
Suppose $V$ is a set of $n$ vertices and $S$ is a subset of $V$ of size $|S| \leq n/2$.
Let \( \mathcal{T}_{S,n}^{(k)} \) denote the subset of $\T_n^{(k)}$ consisting of the $k$-graphs whose edges are all $k$-tuples intersecting both $S$ and $V\setminus S$.
\end{definition}

Note that an equivalent way to define \( \mathcal{T}_{S,n}^{(k)} \) is as all $k$-graphs in $\T_n^{(k)}$ that have a $2$-coloring in which $S,V \setminus S$ are the independent sets.
Observe that a $2$-colorable $k$-graph can belong to several of the sets \( \mathcal{T}_{S,n}^{(k)} \) (e.g. the empty $k$-graph belongs to all of them).
Another observation we will use in all the proofs below is that a uniformly random $k$-graph from \( \mathcal{T}_{S,n}^{(k)} \) can be generated by picking every
$k$-tuple that intersects both $S$ and $V \setminus S$ with probability $1/2$, where these choices are independent. Indeed, the number of $k$-graphs in
\( \mathcal{T}_{S,n}^{(k)} \) is $2^{\binom{n}{k}-\binom{|S|}{k}-\binom{n-|S|}{k}}$ and the above process picks each of these $k$-graphs with probability $(1/2)^{\binom{n}{k}-\binom{|S|}{k}-\binom{n-|S|}{k}}$.
In what follows we use ${\mathcal H} \sim \mathcal{T}_{S,n}^{(k)}$ to denote a uniformly chosen ${\mathcal H}$ from $\mathcal{T}_{S,n}^{(k)}$.

The key step towards the proof of Lemma \ref{good hyper graph probability over all bipartite hypergraphs} is the following lemma, which considers the much simpler
case of $k$-graphs chosen from \( \mathcal{T}_{S,n}^{(k)} \) where $S$ is not too small.

\begin{lemma}\label{good hypergraph probability theorem}
If $n/4 \leq |S| \leq n/2$ then for all large enough $n$, a $k$-graph ${\mathcal H} \sim \mathcal{T}_{S,n}^{(k)}$ is good with probability at least $1-2^{-3n}$.
\end{lemma}

The proof of this lemma will follow from the following three lemmas,
whose proofs make use of the classical McDiarmid Inequality \cite{McDiarmid_1989}.

\begin{fact}[McDiarmid's Inequality \cite{McDiarmid_1989}]\label{fact: McDiarmid's Inequality}
        Let $f:\X_1\times\dots\X_n\rightarrow\mathbb{R}$ satisfy the bounded differences property with bounds $c_1,\dots,c_n$, meaning
        \begin{displaymath}
            \sup_{x_i'\in\X_i}|f(x_1,\dots,x_{i-1},x_i,\dots,x_n)-f(x_1,\dots,x_{i-1}',x_i,\dots,x_n)|\leq c_i\;.
        \end{displaymath}
        Consider the independent random variables $X_1,\dots,X_n$ where $X_i\in\X_i$. Then for any $t>0$
        \begin{displaymath}
            \mathbb{P}\big(f(X_1,\dots,X_n) \leq \mathbb{E}[f(X_1,\dots,X_n)] - t\big)\leq \exp\bigg(-\frac{2t^2}{\sum_{i=1}^nc_i^2}\bigg)\;.
        \end{displaymath}
    \end{fact}

\begin{lemma}\label{lemma: good-hypergraph property i}
Suppose $n/4 \leq |S|\leq n/2$ and $\ell =5k$. Then ${\mathcal H} \sim \mathcal{T}_{S,n}^{(k)}$
satisfies item $(i)$ of Definition \ref{definition:good-hypergraph} with probability at least $1-e^{-n^k/2^{2^{11k}}}$.
\end{lemma}

\begin{proof}
The expected number of copies of $K_{\ell,\ell}$ in $\mathcal{H}\sim\T^{(k)}_{S, n}$ is
\begin{equation}\label{eqcountkll}
        \binom{|S|}{\ell}\binom{n-|S|}{\ell}2^{-2\ell\binom{\ell}{k-1}} \geq \binom{n/4}{\ell}\binom{n/4}{\ell}2^{-2\ell\binom{\ell}{k-1}} \geq \left(\frac{n}{4\ell}\right)^{2\ell}2^{-2\ell\binom{\ell}{k-1}}\geq n^{2\ell}/2^{2^{9k}}.
\end{equation}
    Let $x_e$ be the indicator $1_e(\mathcal{H})$ for all $e\in\binom{n}{k}$, and $\phi:\{0,1\}^{\binom{n}{k}}\rightarrow \mathbb{R}$ be the number of copies of $K_{\ell,\ell}$ in the $k$-graph represented by all $x_e$.
    Each edge is part of at most $\binom{n}{2\ell-k}$ possible copies of $K_{\ell,\ell}$. Hence, in the notation of Fact \ref{fact: McDiarmid's Inequality}, we have $c_e \leq n^{2\ell -k}$ for each variable $x_e$.
    Since there are at most $n^k$ variables, we conclude that $\sum_ec^2_e \leq n^{4\ell -k}$. If the number of copies of $K_{\ell,\ell}$ is less than $n^{2\ell}/2^{2^{10k}}$, then by (\ref{eqcountkll}), this means
    that $\phi$ deviates from its expectation by at least $n^{2\ell}/2^{2^{10k}}$. Hence, by Fact \ref{fact: McDiarmid's Inequality} the probability of this event is at most
\[
    \exp\left(-\frac{(n^{2\ell}/2^{2^{10k}})^2}{n^{4\ell -k}}\right) \leq  e^{-n^k/2^{2^{11k}}}\;.\qedhere
\]
\end{proof}

\begin{lemma}\label{lemma: good-hypergraph property ii}
Suppose $n/4 \leq |S|\leq n/2$ and $\ell = 5k$. Then ${\mathcal H} \sim \mathcal{T}_{S,n}^{(k)}$
satisfies the following with probability at least $1-2^{-3n}$: For every set $A$ of $\ell$ vertices in $S$ there are at least $n/8$ vertices
$v$ in $V \setminus S$ satisfying $N(v,A) \neq \emptyset$. Similarly, for every set $B$ of $\ell$ vertices in $V \setminus S$ there are at least $n/8$ vertices $v$ in $S$ satisfying $N(v,B) \neq \emptyset$.
\end{lemma}

\begin{proof}
Consider any set $B$ in $V \setminus S$ of size $\ell$. A vertex $v$ in $S$ satisfies $N(v,B) \neq \emptyset$ with probability
$(1/2)^{\binom{\ell}{k-1}}$. Hence, the probability that in a set of size $n/8$ all vertices $v$ satisfy $N(v,B) \neq \emptyset$
is $(1/2)^{\frac{n}{8}\binom{\ell}{k-1}} \leq 2^{-5n}$, where we have used the independence of the edges contributing to $N(v,B)$. Therefore, the probability that there is such a set of size $n/8$ is at most $2^{-4n}$.
Hence, with probability at least $1-2^{-4n}$ there are at least $n/4-n/8=n/8$ vertices satisfying $N(v,B) \neq \emptyset$. The same argument applies to every fixed set $A$ in $S$ of size $\ell$ (with even higher probability since $|V \setminus S| \geq |S|$). Since there are at most $\binom{n}{2\ell}$ choices for the sets $A,B$ we conclude that with probability at least $1-2^{-3n}$ the condition holds for all $A,B$.
\end{proof}

\begin{lemma}\label{lemma: good-hypergraph property iii}
Suppose $n/4 \leq |S|\leq n/2$ and $\ell = 5k$. Then ${\mathcal H} \sim \mathcal{T}_{S,n}^{(k)}$
satisfies the following with probability at least $1-2^{-3n}$: for every $u \in S$ (resp. $u \in V \setminus S$) and every set $X \subseteq V \setminus S$ (resp. $S$) of size $n/8$ we have $|N(u,X)|\geq n^{k-1}/k^{2k}$.
\end{lemma}

\begin{proof}
Consider $u \in S$ and $X \subseteq V \setminus S$ of size $n/8$. Note that we can think of
$|N(u,X)|$ as a sum of $\binom{n/8}{k-1}$ binary indicator random variables. Hence, in the setting of Fact \ref{fact: McDiarmid's Inequality}
we have $\binom{n/8}{k-1}$ variables, each with $c_i=1$. Then the expected size of $N(u,X)$ is $\frac12 \binom{n/8}{k-1} \geq n^{k-1}/k^{3k}$.
So by Fact \ref{fact: McDiarmid's Inequality} the probability that $|N(u,X)| \leq n^{k-1}/k^{4k}$ is at most
$$
e^{-\frac{(n^{k-1}/k^{4k})^2}{n^{k-1}}} \leq e^{-n^{k-1}/k^{8k}}\;.
$$
Since there are at most $2^{2n}$ choices for $u$ and $X$, we get that the lemma holds with probability at least $1-2^{-n^{k-1}/k^{9k}}$.
\end{proof}

\begin{proof}[Proof (of Lemma \ref{good hypergraph probability theorem})]
Suppose $S\subset V$ satisfies $n/4 \leq |S|\leq n/2$ and $\mathcal{H}\sim\T^{(k)}_{S, n}$. Let $\ell = 5k$ be as in (\ref{l constant function}). The probability that either Lemma \ref{lemma: good-hypergraph property i}, \ref{lemma: good-hypergraph property ii} or \ref{lemma: good-hypergraph property iii} fails is at most $$e^{-n^k/2^{2^{11k}}} + 2^{-5n} + 2^{-n^{k-1}/k^{9k}}< 2^{-2n}$$ for large enough $n$. It thus remains to prove
that any $\mathcal{H}$ satisfying the assertions of these three lemmas is good. Indeed, Lemma \ref{lemma: good-hypergraph property i} guarantees that item $(i)$ of Definition \ref{definition:good-hypergraph} holds. As to item $(ii)$, take any copy $K$ of $K_{\ell,\ell}$ in $\mathcal{H}$.
Since $\ell \geq 2k-3$ we know by Claim \ref{uniquecoloirng} that $K$ has a unique $2$-coloring. Since by construction $S,V \setminus S$ is a legal $2$-coloring of $\mathcal{H}$
and thus of $K$, we deduce that the independent sets $A,B$ of $K$ are contained in $S,V \setminus S$.
If Lemma \ref{lemma: good-hypergraph property ii} holds, then the sets $N_A,B_B$ in item $(ii)$ of Definition \ref{definition:good-hypergraph} are each of size at least $n/8$. Applying Lemma \ref{lemma: good-hypergraph property ii} to $N_A,N_B$
guarantees that item $(ii)$ of Definition \ref{definition:good-hypergraph} holds.
\end{proof}

\begin{remark}
The algorithms given in Sections \ref{coloring algorithm section} and \ref{oracle section} additionally apply for $\T^{(k)}_{S,n}$, and any other distribution over $2$-colorable $k$-graphs for which $\mathcal{H}$ drawn from the distribution is a good $k$-graph with probability at least $1 - 2^{-2n}$. This is because the guarantees of the analysis of the correctness and runtime of the algorithms only rely on the $k$-graph being good with high probability.
\end{remark}

Given Lemma \ref{good hypergraph probability theorem}, we now prove Lemma \ref{good hyper graph probability over all bipartite hypergraphs}.

\begin{proof}[Proof (of Lemma \ref{good hyper graph probability over all bipartite hypergraphs})]
Let $V$ be a set of $n$ vertices and $\T_n^{(k)}$ and ${\cal T}_{S,n}^{(k)}$ as defined above. Let ${\cal B}_n^{(k)}$ denote the
subset of $\T_n^{(k)}$ consisting of the $k$-graphs that are not good. Our goal in thus to prove that $|{\cal B}_n^{(k)}| \leq |\T_n^{(k)}|/2^{2n}$. Let ${\cal B}_{S,n}^{(k)}$ denote the
subset of ${\cal T}_{S,n}^{(k)}$ consisting of the $k$-graphs that are not good. We will show that if $|S| \leq n/4$ then
\begin{equation}\label{eqBvsT}
|{\cal B}_{S,n}^{(k)}| \leq |\T_n^{(k)}|/2^{3n}\;,
\end{equation}
but let us first see how to derive the lemma from this bound. Indeed,
\begin{align*}
|{\cal B}_n^{(k)}| &\leq \sum_{S \subseteq V: |S| \leq n/2}|{\cal B}_{S,n}^{(k)}|\\
&=\sum_{S \subseteq V: |S| \leq n/4}|{\cal B}_{S,n}^{(k)}|+\sum_{S \subseteq V: n/4 \leq |S| \leq n/2}|{\cal B}_{S,n}^{(k)}|\\
& \leq \sum_{S \subseteq V: |S| \leq n/4}|{\cal T}_{n}^{(k)}|/2^{3n}+\sum_{S \subseteq V: n/4 \leq |S| \leq n/2}|{\cal B}_{S,n}^{(k)}|\\
& \leq \sum_{S \subseteq V: |S| \leq n/4}|{\cal T}_{n}^{(k)}|/2^{3n}+\sum_{S \subseteq V: n/4 \leq |S| \leq n/2}|{\cal T}_{S,n}^{(k)}|/2^{3n}\\
& \leq \sum_{S \subseteq V: |S| \leq n/4}|{\cal T}_{n}^{(k)}|/2^{3n}+\sum_{S \subseteq V: n/4 \leq |S| \leq n/2}|{\cal T}_{n}^{(k)}|/2^{3n}\\
& \leq |{\cal T}_{n}^{(k)}|/2^{2n}\;,
\end{align*}
where in the second inequality we use (\ref{eqBvsT}) and in the third inequality we use Lemma \ref{good hypergraph probability theorem}.

We now prove (\ref{eqBvsT}). Fix $S \subseteq V$ of size $n/2$. Then we clearly have
$$
|\T_n^{(k)}| \geq |\T_{S,n}^{(k)}|=2^{\binom{n}{k}-2\binom{n/2}{k}}.
$$
Similarly, for every $|S| \leq n/4$ we have
$$
|{\cal B}_{S,n}^{(k)}|= 2^{\binom{n}{k}-\binom{|S|}{k}-\binom{n-|S|}{k}} \leq 2^{\binom{n}{k}-\binom{3n/4}{k}}\;.
$$
So to prove (\ref{eqBvsT}) we just need to show that $\binom{3n/4}{k} \geq 2\binom{n/2}{k} + 3n$ for all large enough $n$. Indeed,
$$
\binom{3n/4}{k} \geq \binom{n/2}{k} + \frac{n}{4}\binom{n/2}{k-1} \geq \binom{n/2}{k} + \frac{n}{4}\binom{n/2-1}{k-1}=\binom{n/2}{k} + \frac{2k}{4}\binom{n/2}{k}\;.
$$
Since $k \geq 3$, the RHS is at least $\frac52 \binom{n/2}{k}$ which is at least $2\binom{n/2}{k} + 3n$ for large enough $n$, thus completing the proof.
\end{proof}

\bibliography{bibliography}

\appendix

\section{Average-Case LCA for Hypergraph Coloring}\label{LCA partition oracle section}
In the next section, we present our \textit{coloring oracle} result within the context of the average-case Local Computation Algorithm (LCA) model \cite{beyond_worst_case_LCA}, proving Theorem \ref{LCA partition oracle}. The LCA model provides a general framework for efficiently accessing parts of a global solution without computing it in its entirety. The coloring oracle model can be viewed as a specific type of LCA, though there are some differences in the specifications of the average-case model that we will address in this section.


We begin by formally defining the LCA model, then introduce the recent development of the average-case LCA model, which is relevant to our coloring oracle setting. Finally, we analyze our algorithm within the framework of the average-case LCA model.

\begin{definition}\label{LCA definition}
    A \textbf{Local Computation Algorithm (LCA)} for a problem $\Pi$ is an oracle $\mathcal{O}$ with the following properties. $\mathcal{O}$ is given probe access to input $G$, a sequence of random words\footnote{A random word is a random entry in $[n]$, which can equivalently be thought of as $O(\log n)$ bits of randomness received at once.} $\vec{r}$, and a local memory. For any query $q$ in a family of admissible queries to the output, $\mathcal{O}$ uses only its oracle access to $G$ (i.e. \textit{probes} to $G$), random words from $\vec{r}$, and local memory to answer the query $q$. After answering the query, $\mathcal{O}$ erases its local memory, including the query $q$ and its response.

    For an input $G$ and query $q$, let $PC_{\mathcal{O}}(G, q)$ be the expected number of probes (over the random choice of $\vec{r}$) the LCA makes to answer the query $q$ on the input $G$. Let $PC_{\mathcal{O}}(G) = \max_{q} PC_{\mathcal{O}}(G, q)$ be the maximum of the expected probe complexity over all queries to input $G$. Finally, the \textbf{probe complexity} $PC_{\mathcal{O}}(n)$ is the maximum of $PC_{\mathcal{O}}(G) $ over all inputs $G$ parameterized by size $n$.

    We say $\mathcal{O}$ has \textbf{expected runtime} $t(n)$ if its runtime on any input graph and any query is at most $t(n)$ in expectation (over the random choice of $\vec{r}$).
\end{definition}


We will show how our result fits in the framework of average-case LCAs, which are defined below. The average-case LCA model was introduced to address the natural question of whether more efficient local computation algorithms can be designed under the assumption that the input is drawn from some distribution. In this setting, the goal is for the LCA to succeed with high probability over a randomly chosen input from the distribution, instead of considering the LCA's performance only over worst-case graphs.


\begin{definition}[\cite{beyond_worst_case_LCA}]\label{def:avg-case-lca}
    An oracle $\mathcal{O}$ for a problem $\Pi$ is a \textbf{average-case local computation algorithm} for a distribution $\mathcal{G}$ over objects if, with probability $(1 - \frac{1}{n})$ over $G\leftarrow\mathcal{G}$, $\mathcal{O}$ is an LCA for $\Pi$ on $G$.

    We say that average-case LCA $\mathcal{O}$ has \textbf{worst-case probe complexity} $PC(n)$ if the maximum $PC_{\mathcal{O}}(G) $ over $G \leftarrow \mathcal{G}$ is $PC(n)$, where $\mathcal{G}$ is a distribution over inputs $G$ parameterized by size $n$.
\end{definition}

We note that, while the coloring oracle needed to return a consistent coloring on any input $k$-uniform bipartite hypergraph $\mathcal{H}$, the average-case LCA can fail to return a consistent coloring on a $1/n$ fraction of the input. This allows us to make a straightforward modification (Step 7 in Algorithm 3) to obtain a slightly simpler algorithm for the average-case LCA, and to achieve better guarantees about the probe complexity (worst-case instead of average-case).

\begin{algorithm}[H]\label{alg:avg-lca}
\caption{Average-case LCA for Bipartite $k$-Graphs}
    \textbf{Input:} Vertex $u$ in some bipartite $k$-graph $\mathcal{H}$, oracle access to $E(\mathcal{H})$, and random tape $\Vec{r}$ of words. \\
    \textbf{Output:} Color assignment to $u$.\\

    1.\quad\textbf {Procedure:} \textsc{Coloring-LCA($\mathcal{H}$,$u$)}\\
    2.\quad Use $\Vec{r}$ to uniformly sample $(2\ell+k-1)$-tuples $(x_1,\ldots,x_{2\ell},y_1,\ldots,y_{k-1})$ of vertices until $(i)$ and $(ii)$ hold:\\
    $~~~~~~~$ $(i)$ Vertices $x_1,\ldots,x_{2\ell}$ span a copy of $K_{\ell,\ell}$. Set $A,B$ to be the independent sets of this copy.\\
    $~~~~~~~$ $(ii)$ Vertices $y_1,\ldots,y_{k-1}$ satisfy one of the following:\\
    $~~~~~~~~$ $(ii.1):$ $(1,y_1,\dots,y_{k-1})\in E(\mathcal{H})$ and $N(y_i,A)\neq \emptyset$ for every $i\in [k-1]$. Set $c_A=0,c_B=1$\\
    $~~~~~~~~$ $(ii.2):$ $(1,y_1,\dots,y_{k-1})\in E(\mathcal{H})$ and $N(y_i,B)\neq \emptyset$ for every $i\in [k-1]$. Set $c_A=1,c_B=0$\\
    $~~~~~~~$ \textbf{If} all $(2\ell+k-1)$-tuples were inspected, and none of them satisfies $(i),(ii)$, goto Step 7.\\
    3.\quad Use $\Vec{r}$ to uniformly sample $(k-1)$-tuples of vertices $v_1,\dots,v_{k-1}$:\\
    4.\qquad \textbf{If} $(u,v_1,\dots,v_{k-1})\in E(\mathcal{H})$ and $N(v_i,A)\neq \emptyset$ for every $i\in [k-1]$ then return $c_A$.\\
    5.\qquad \textbf{If} $(u,v_1,\dots,v_{k-1})\in E(\mathcal{H})$ and $N(v_i,B)\neq \emptyset$ for every $i\in [k-1]$ then return $c_B$.\\
    6.\qquad \textbf{If} all $(k-1)$-tuples have been inspected, and none of them satisfies $4-5$, then goto Step 7.\\
    7.\quad Return $0$ (arbitrarily color).
\end{algorithm}

\begin{proof}[Idea of Proof of Theorem \ref{LCA partition oracle}]
    In the average-case LCA model, we need that for hypergraphs $G$ sampled from the uniform distribution over $k$-uniform bipartite hypergraphs, with probability at least $1 - \frac{1}{n}$ for $G$ over this distribution, Algorithm 3 is an LCA for 2-coloring on $G$. When $G$ is a \textit{good} hypergraph, the algorithm succeeds at providing local access to a $2$-coloring on $G$. Since the probability that a random $k$-uniform bipartite hypergraph is \textit{not} good is exponentially small, the algorithm satisfies this error guarantee. Indeed, for $k = 2$, the proofs in Section \ref{section:good-hypergraph-proofs} can be used to show that a random bipartite $2$-graph is \textit{not} good with probability less than $1/n$, so the error guarantee is also satisfied in this case.

    Therefore, since Algorithm 2 satisfied the definition of a \emph{coloring oracle}, Algorithm 3 (which only differs by simplifying Step 7) is an average-case LCA. The probe complexity and runtime are also $O(1)$ for Algorithm 3 since $O(1)$ words are sampled, and for each sample only operations of cost $O(1)$ are used (as in the proof for Algorithm 2).

    Finally, we remark that as in the coloring oracle, the LCA uses independent (i.e. not shared) random coins for each query. Since the LCA model does provision shared randomness, we simply have the LCA divide the shared random tape into blocks $r_1,\ldots,r_n$ of sufficiently long length, and use the section of the random tape $r_u$ on query $u$.
\end{proof}

\end{document}